\documentclass[12pt]{article}

\usepackage[T1]{fontenc}
\usepackage{authblk}
\usepackage[titletoc]{appendix}
\usepackage[dvips]{graphicx}
\usepackage{amssymb}
\usepackage{amsmath}
\usepackage{amsfonts}
\usepackage{color}
\usepackage{amssymb}
\usepackage{ntheorem}
\usepackage{array}
\usepackage{lipsum}
\usepackage[labelsep=none]{caption}
\usepackage[mathscr]{euscript}
\usepackage[shortlabels]{enumitem}
\newtheorem{lemma}{{\sc Lemma}}

\newtheorem{cor}{{\sc Corollary}}
\newtheorem{theorem}{{\sc Theorem}}

\newtheorem{example}{{\sc Example}}

\newtheorem{definition}{{\sc Definition}}

\newtheorem{remark}{{\sc Remark}}
\newenvironment{proof}{\noindent {\bf \sl Proof\/}:\enspace}
{\hfill $\blacksquare{}$ \vspace{12pt}}

\newenvironment{proofs}{\noindent {\bf \sl Proof\/}:\enspace}
{\hfill $\square{}$ \vspace{12pt}}

\title{\bf {\Large Stability in Many-to-One Matching with Couples having Responsive Preferences}\footnote{This work was made possible through the invaluable guidance and constructive feedback provided by Dr. Hans Peters, to whom the authors express their deepest gratitude.}}

\author[1]{Shashwat Khare}
\author[2]{Souvik Roy\footnote{Corresponding author: souvik.2004@gmail.com}}
\author[3]{Ton Storcken}
\affil[1]{Independent Researcher}
\affil[2]{Statistical Sciences Division, Indian Statistical Institute}
\affil[3]{Quantitative Economics, Maastricht University}

\begin{document}
	
	\maketitle
	
	\begin{abstract}
		\noindent This paper studies matching markets where institutions are matched with possibly more than one individual. The matching market contains some couples who view the pair of jobs as complements. First, we show by means of an example that a stable matching may fail to exist even when both couples and institutions have responsive preferences. Next, we provide conditions on couples' preferences that are necessary and sufficient to ensure a stable matching for every preference profile where institutions may have any responsive preference. Finally, we do the same with respect to institutions' preferences, that is, we provide conditions on institutions' preferences that are necessary and sufficient to ensure a stable matching for every preference profile where couples may have any responsive preference.

	\end{abstract}

	\section{Introduction}
	\subsection{Background of the problem}
A substantial literature has developed on various market designs aimed at identifying an ``optimal'' matching procedure in labor markets. In many centralized labor markets, stability is a crucial condition for optimality. A matching is said to be stable if there exist no institution-individual pairs who are not matched to each other but would both strictly prefer to be matched together over their current allocation. Such pairs are referred to as \emph{blocking pairs}. For labor markets consisting of institutions (hereafter referred to as hospitals) and individuals (hereafter called doctors), stable matchings are known to exist under suitable domain restrictions.

Roth \cite{roth1984labor} was the first to highlight that the presence of couples in labor markets may preclude the existence of a stable matching. This phenomenon arises because couples may perceive pairs of jobs as complements, thus violating the assumption of independence in individual choices.

Subsequently, Kelso and Crawford \cite{kelso1982job}, Roth \cite{roth1985college}, Alkan and Gale \cite{alkan2003revealed}, and Hatfield and Milgrom \cite{hatfield2005matching} demonstrated that a sufficient degree of substitutability in preferences guarantees the existence of a stable matching. These works assume substitutability in hospitals' preferences over sets of doctors. Later, Klaus and Klijn \cite{klaus2005responsive} introduced the assumption of responsiveness in couples' preferences over ordered pairs of hospitals and showed that under such conditions, stable matchings exist. Responsiveness entails that unilateral improvements in the preference of one member of the couple translate into an improved outcome for the couple as a whole.
	
	\subsection{Motivation of the problem} As mentioned, we assume that the institutions are hospitals and the individuals are doctors, where a subset of doctors consists of couples. We assume that couples' preferences are responsive over pairs of hospitals, and hospitals' preferences are responsive over sets of doctors.

It is noteworthy that, given the individual preferences of the couple members, more than one (joint) couple preference may satisfy responsiveness. For example, consider a couple \( c = \{f,m\} \) and two hospitals \( h_1 \) and \( h_2 \) such that \( f \) prefers \( h_1 \) to \( h_2 \), but \( m \) prefers \( h_2 \) to \( h_1 \). Suppose further that \((x,y)\) denotes an allocation where \( f \) is matched with \( x \) and \( m \) is matched with \( y \). By responsiveness, we know that \((h_1,h_2)\) is preferred to both \((h_1,h_1)\) and \((h_2,h_2)\), whereas both \((h_1,h_1)\) and \((h_2,h_2)\) are preferred to \((h_2,h_1)\). However, responsiveness does not impose an ordering between the pairs \((h_1,h_1)\) and \((h_2,h_2)\). Thus, multiple complete responsive extensions of couple preferences exist. Similarly, given a hospital's preference over doctors, more than one responsive preference over subsets of doctors can arise.

Allowing for complete responsive preferences for both couples and hospitals is not innocuous; in such cases, stable matchings may fail to exist. On the other hand, permitting comparability between pairs such as \((h_1,h_1)\) and \((h_2,h_2)\) in the example above seems natural. Likewise, it appears reasonable that hospitals admitting teams of doctors compare, for instance, two pairs of doctors where one pair contains an A-ranked and a C-ranked doctor, and the other pair contains two B-ranked doctors. As we will see, allowing such comparisons influences the stability of matchings. 

However, it is important to note that such comparisons might not always be possible, even when hospitals’ preferences satisfy responsiveness. For example, by standard replication of hospitals up to their capacities to transform the problem into a one-to-one matching, these pairs of doctors are effectively considered incomparable. This occurs because, from the hospital's perspective, a B-ranked doctor is better than a C-ranked doctor but worse than an A-ranked doctor. Consequently, the replication approach implies that these two sets of doctors are incomparable. 

The possibility of allowing all such comparisons is the primary reason why the results presented here differ from those of Klaus and Klijn \cite{klaus2005responsive}, who established the existence of stable matchings for every case of couples’ responsive preferences.

	\subsection{Contribution of the paper} 
In this paper, we formalize the notion of complete responsive preferences for both hospitals and couples and examine its implications for the existence of stable matchings.

First, we demonstrate by example that stable matchings may fail to exist under arbitrary responsive preference profiles. Next, we introduce a condition on couples' preferences, termed \emph{extreme-altruism}, which is necessary and sufficient to guarantee the existence of a stable matching for every responsive extension of hospitals’ preferences over sets of doctors. To illustrate this condition, consider \( k \) hospitals \( \{h_1, \ldots, h_k\} \) and a couple. Suppose both members of the couple strictly prefer \( h_i \) to \( h_j \) for all \( 1 \leq i < j \leq k \). Then, extreme-altruism requires that for any \( 1 \leq i < j \leq k \), the couple prefers the allocation where one member is matched with \( h_i \) and the other with \( h_k \) over the allocation where both are matched with \( h_j \).

Subsequently, we provide necessary and sufficient conditions on hospitals’ preferences that ensure the existence of stable matchings for every responsive extension of couples’ preferences. We refer to this condition as \emph{aversion to couple diversity}. To explain this, consider a hospital \( h \) with a preference order \( P_h \) over individual doctors, a couple \( c = \{f,m\} \), and two other doctors \( d, d' \) (possibly another couple). Suppose \( f P_h d P_h d' P_h m \). As previously discussed, responsiveness imposes no restrictions on \( h \)'s relative preference between the sets \( \{f,m\} \) and \( \{d,d'\} \). The aversion to couple diversity condition states that in such cases, the hospital prefers \( \{d,d'\} \) to \( \{f,m\} \). Informally, this condition asserts that a hospital disfavors employing couples whose members are relatively more dissimilar with respect to its ranking of individual doctors.

Another significant contribution of this paper is that, beyond characterizing preference profiles that guarantee stable matchings, we also provide algorithms that produce a stable matching whenever one exists.

Thus, we believe our results offer a more comprehensive understanding of the existence of stable matchings when hospitals and couples have responsive preferences, thereby complementing the work of Klaus and Klijn \cite{klaus2005responsive}.
		
\subsection{Organization of the paper}
In the next section, we present the formal framework of our model and provide all necessary definitions. We also introduce and describe an algorithm for matching doctors with hospitals, which will be employed throughout this paper. 

In Section~\ref{section3}, we provide an example demonstrating that the existence of a stable matching is not guaranteed under arbitrary responsive preference profiles. 
Section~\ref{sec3} establishes conditions on couples' preferences that are necessary and sufficient to guarantee the existence of a stable matching for arbitrary responsive extensions of hospitals' preferences. 
In Section~\ref{sec4}, we state conditions on hospitals' preferences that are necessary and sufficient to ensure the existence of a stable matching for arbitrary responsive extensions of couples' preferences. 
We conclude in Section~\ref{sec5} by presenting a formal example that illustrates the distinction between our results and those of Klaus and Klijn~\cite{klaus2005responsive}.


	\section{The framework}
	We consider many-to-one matchings between doctors and hospitals. We denote by $H$ a finite set of hospitals. Each hospital $h \in H$ has a finite capacity, denoted by $\kappa_h$.
	
	We denote by $D$ a finite set of doctors. We assume that $D=M \cup F \cup S$ where $F,M,S$ are pairwise disjoint sets with $|M|=|F|$. Here, the doctors in $F$ and $M$ together form fixed couples. Also, the doctors in $S$ are those who are not part of any couple. We  denote the set of couples by $C=\{\{f_1,m_1\},\{f_2,m_2\},\ldots\}$ and a generic couple by $c=\{f,m\}$. We denote by $\lambda \notin H$, a dummy hospital which we use to represent a doctor being unemployed.

	Throughout this paper, we assume $|H| \geq 2$, $|D|\geq 4$, $|C| \geq 1$, and $\kappa_h \geq 2$ for all $h \in H$ and $\kappa_{\lambda}=|D|$.
		
	For notational convenience, we do not always use braces for denoting singleton hospitals, doctors or couples.

	\subsection{Matching}
	
	\begin{definition}
		A matching $\mu$ is a correspondence from $H \cup \{\lambda\}$ to $D$ such that for all $h \in H$, $|\mu(h)| \leq \kappa_{h}$. Moreover $\mu(h_1) \cap \mu(h_2) = \emptyset$ for any $h_1,h_2 \in H$ with $h_1 \neq h_2$.
	\end{definition}
	
	For ease of notation, whenever $d \in \mu(h)$ for some $d \in D$ and $h \in H$,  we write  $\mu(d)=h$. We say that a doctor is matched with $\lambda$ to mean that the doctor is unemployed. More formally, if $d \notin \mu(h)$ for all $h \in H$, then  $\mu(d)=\lambda$. For a couple $c=\{f,m\} \in C$ and for hospitals $h,h' \in H \cup \{\lambda \}$, we write $\mu(c)=(h,h')$ to mean  $\mu(f)=h$ and $\mu(m)=h'$. 
	Further, for  a hospital $h \in H$ and a matching $\mu$, we say $h$ has $\kappa_h - |\mu(h)|$ vacant positions at $\mu$.
	
	\subsection{Preferences}
	In this section, we introduce the notion of preferences of doctors and hospitals, and present some restrictions on them. 
	
	For a set $X$, we denote by $\mathbb{L}(X)$ the set of linear orders on $X$, i.e., complete, reflexive, transitive, and antisymmetric binary relations over $X$. An element of $\mathbb{L}(X)$ is called a preference (over $X$). For any $i \in H \cup D \cup C$, $R_i$ denotes a preference of $i$ and $P_i$ denotes its strict part. Since a preference is antisymmetric $x R_i y$ implies either $x = y$ or $x P_i y$. We say $x$ is weakly preferred to $y$ to mean  $xR_iy$, and $x$ is (strictly) preferred to $y$ to mean $xP_iy$. For $P_i \in \mathbb{L}(X)$ and $k \leq |X|$, we define the $k$-th ranked element in $P_i$, denoted by $r_k(P_i)$, as follows:  $r_k(P_i)= x \in X$ if $|\{y \in X: yR_i x\}|=k$.
	
\subsubsection{Preferences of hospitals} 
For any hospital $h \in H$, let $D_h$ be the set of acceptable doctors. A preference of hospital $h$, denoted by $\bar{P}_h$, is a linear order over $D_h$. Thus, $\bar{P}_h \in \mathbb{L}(D_h)$. A hospital prefers to have any doctor from this set of acceptable doctors, over having a vacant spot. Similarly, a hospital prefers to have a vacant spot to having doctors which do not belong to the set of acceptable doctors. We assume that the dummy hospital $\lambda$ finds all doctors acceptable.	Thus, $D_{\lambda}=D$. Also, $\lambda$ is indifferent between all doctors.
	
	For any hospital $h \in H$, a preference $\bar{P}_h$ over individual doctors is extended  to a preference $P_h$ over feasible subsets of acceptable doctors $\{D' \subseteq D_h : |D'| \leq \kappa_h \}$. 
		
	\begin{definition}
		We say $P_h \in \mathbb{L}(\{D' \subseteq D_h : |D'| \leq \kappa_h \})$ is \emph{responsive} if	
		\begin{itemize}
			\item[(i)] for all $D' \subseteq D_h$ with $D' \neq \emptyset$ and $|D'| \leq \kappa_h$, $D' P_h \emptyset$.
			\item[(ii)]  for all $d,d' \in D_h$, $\{d\} R_h \{d'\}$ if and only if $d \bar{R}_h d'$, and
			
			\item[(iii)] for all $D',D'' \subseteq D_h$ with $|D'| < \kappa_h$, $|D''| < \kappa_h$ and all $d \in  D_h \setminus (D'\cup D'')$,  $(D' \cup \{d\}) P_h (D'' \cup \{d\})$ if and only if $D' P_h D''$.
		\end{itemize}
		
	\end{definition}
	
	Having define $P_h$ over all feasible subsets of acceptable doctors, we extend this preference over the set of all feasible subsets of doctors $\{D' \subseteq D: |D'| \leq \kappa_h \}$.
	
	\begin{definition}
	For all $D' \subseteq D$ such that $|D'| \leq \kappa_h$ and $D' \not\subset D_h$, we have $\emptyset P_h D'$.
	\end{definition}

		For notational convenience, for any hospital $h \in H$, any couple $c=\{f,m\} \in C$ and any doctor $d \in D \setminus \{f,m\}$, $d P_h c$ means $d P_h f$ and $d P_h m$.

	\subsubsection{Preferences of doctors and couples}
	A preference of a doctor $d \in D$, denoted by $P_d$, is an element of $\mathbb{L}(H \cup \{\lambda\})$.	A preference of a couple $c=\{f,m\} \in C$, denoted by $P_c$, is an element of $\mathbb{L}((H \cup \{\lambda\})^2)$.  We call a preference of a couple responsive if a unilateral improvement in the position of one member of the couple is beneficial for the couple.
	
	\begin{definition}\label{def1} Let $c=\{f,m\} \in C$ be a couple. Let $P_{f}$ be a preference of $f$ and $P_m$ be a preference of $m$. A preference $P_c \in \mathbb{L}({(H \cup \{\lambda\})}^2)$ of the couple $c$ is called responsive (with respect to $P_f$ and $P_m$) if for all $h,h_1,h_2 \in H \cup \{\lambda\}$, we have
		\begin{enumerate}
			\item[(i)] $(h_1,h)P_c(h_2,h)$ if and only if $h_1 P_f h_2$, and 
			\item[(ii)] $(h,h_1)P_c(h,h_2)$ if and only if $h_1 P_m h_2$.			 
		\end{enumerate}
	\end{definition}
	
For any $c=\{f,m\} \in C$, a responsive preference $P_c$ induces unique marginal preferences $P_f$ and $P_m$ for $f$ and $m$ respectively. 
	
	\subsubsection{Preference profiles and matching problems}
	 A preference profile is a collection of responsive preferences for all hospitals in $H$, all doctors in $D$ and all couples in $C$. Thus, a preference profile $P$ is a tuple of preferences $(\{P_d\}_{d\in D},\{P_c\}_{c \in C},\{P_h\}_{h \in H})$, where for all $d \in D$, $c \in C$ and  $h \in H$, $P_d$ is a preference of doctor $d$, $P_c$ is a responsive preference of couple $c$, and $P_h$ is a responsive preference of hospital $h$ over acceptable and feasible sets of doctors, respectively. Note that, for any hospital $h \in H$, $D_h$ is an inherent part of $P_h$. This means that a preference $P_h$ automatically specifies the acceptable set $D_h$ of the hospital.
	 	
	A matching problem is a tuple consisting of a set of hospitals with corresponding capacities, a set of doctors with its partition into $F$, $M$, $S$, and a corresponding preference profile.

	\subsection{Stability}
There are different notions of stability based on different types of permissible blocking coalitions.
	
	Let $\mu$ be a matching and $P$ be a preference profile. We say a hospital $h$ \emph{prefers} to have a set of doctors $D'$ (possibly empty) to a subset of doctors in $\mu(h)$  if there is $D'' \subseteq \mu(h)$ with $D' \cap D'' = \emptyset$ such that $\{(\mu(h) \setminus D'' )\cup D'\}P_{h}\mu(h)$. Similarly, we say a doctor $d$ (or a couple $c$) \emph{prefers} a hospital $h$ to $\mu(d)$ (or a pair of hospitals $(h,h')$ to $\mu(c)$)  if $h P_d \mu(h)$ (or $(h,h')P_c\mu(c)$). Note that if a hospital prefers a set of doctors to its assignment at $\mu$, then by definition, that hospital is not matched with any of those doctors  at $\mu$. Moreover, it could also be that $h$ has some unacceptable doctors $D''$ in $\mu(h)$, thus $h$ prefers $D'=\emptyset$ to $\mu(h)$. 
	
	Similarly, if a doctor (or a couple) prefers a hospital (or a pair of hospitals) to its assignment at $\mu$, then that doctor (or at least one member of that couple) is not matched with the hospital (or the corresponding hospital) at $\mu$. 
	
    Now, we define the notion of blocking. Note that, since $\lambda$ is indifferent between all sets of doctors, and $D_{\lambda}=D$, thus $\lambda$ always prefers to have any doctor than not having that doctor. First, we introduce the notion of blocking between a hospital and a doctor in $S$. 
    
	\begin{definition}
		Let $s \in S$, $h \in H \cup \{\lambda\}$ and let $\mu$ be a matching. Then $(h,s)$ \emph{blocks} $\mu$ if $h$ prefers $s$ to $\mu(h)$ and $s$ prefers $h$ to $\mu(s)$.
	\end{definition}

	Next, we define the notion of blocking between a pair of hospitals and a couple.
	
	\begin{definition}\label{def5}
		Let $\mu$ be a matching and let $c=\{f,m\} \in C$ and $(h_f,h_m) \in {(H \cup \{\lambda\})^2}$. Then, $((h_f,h_m),c)$ blocks $\mu$ if $c$ prefers $(h_f,h_m)$ to $\mu(c)$ and 
		\begin{enumerate}
			\item[(i)] if $h_f \neq h_m$ and $\mu(f) \neq h_f$, then $h_f$ prefers $f$ to $\mu(h_f)$, 
			\item[(ii)] if $h_f \neq h_m$ and $\mu(m) \neq h_m$, then $h_m$ prefers $m$ to $\mu(h_m)$, and
			\item[(iii)] if $h_f = h_m$, then $h_f$ prefers $\{f,m\}$ to $\mu(h_f)$.
		\end{enumerate}	
		
	\end{definition}

   It is worth mentioning that the blocking notion takes complementarity of a couple being accepted into account (by allowing the notion of a hospital
being interested in a couple) but it does not take the couple into account when accepting single doctors and possibly removing members of a couple. In other words, there is an asymmetry here.

We consider this asymmetry in our model since it is not practical for big institutions like hospitals to consider the possibility of losing a member of a couple while removing the other member. This is because this possibility depends on factors like which hospital the removed member will join, whether the couple prefers to be together in that hospital, etc. Clearly, such situations can only be modeled by using a farsighted notion of blocking, which would complicate the model considerably. 
    
     Thus, by allowing the notion of a hospital being interested in a couple, the blocking definition takes complementarity of a couple being accepted into account. However, the hospital does not take this into account while accepting single doctors at the cost of removing a member of the couple from the hospital. The asymmetry arising here is the main reason, why the results obtained in this chapter are different to the results obtained by choice function approach in many-to-many matchings. \footnote{See Konishi and \"{U}nver\cite{konishi2006groups}, Echenique and Oviedo\cite{echenique2006stability}, Hatfield and Kojima\cite{hatfield2010substitutes}, for notions of stability in many-to-many matchings.}

	Whenever a matching $\mu$ is blocked by  $((h_f,\mu(m)),c)$  for some $c =\{f,m\} \in C$ and some $h_f \in H \cup \{\lambda\}$, for ease of presentation we say that $\mu$ is blocked  by $(h_f,f)$. Similarly we say that $\mu$ is blocked  by $(h_m,m)$ if $\mu$ is blocked by  $((\mu(f),h_m),c)$.

Our next remark follows from the responsiveness of couples' preferences and Definition \ref{def5}.
	
	\begin{remark}
		Let $\mu$ be a matching and $c=\{f,m\}$ be a couple. Suppose for some $x \in \{f,m\}$ and some hospital $h_x \in H$, we have $h_x P_x \mu(x)$ and $((\mu(h_x) \setminus d) \cup \{x\}) P_{h_x} \mu(h_x)$ for some $d \in \mu(h_x) \setminus \{f,m\}$, then $(h_x,x)$ blocks $\mu$. 
	\end{remark}	
	
	A matching is stable if it cannot be blocked by any blocking pair. More formally, we get the following definition.
	
	\begin{definition}
		A matching $\mu$ is stable, if 
		\begin{itemize}
			\item[(i)] for all $h \in H \cup {\lambda}$ and $s \in S$, $(h,s)$ does not block $\mu$,
			\item[(ii)] for all $(h_f,h_m) \in (H \cup {\lambda})^2$ and $c \in C$, $((h_f,h_m),c)$ does not block $\mu$, and
			\item[(iii)] for all $h \in H$, $(h,\emptyset)$ does not block $\mu(h)$, i.e., $h$ does not prefer $\emptyset$ to $\mu(h)$.
		\end{itemize} 
	\end{definition}
	
	Now, we define the concept of individual rationality. 
	\begin{definition}
		A matching $\mu$ is individually rational if 
		\begin{enumerate} 
			\item[(i)] for all $s \in S$, $\mu(s) R_s \lambda$,
			\item[(ii)] for all $c \in C$, $\mu(c) R_c (\lambda,\lambda)$, and 
			\item[(iii)] for all $h \in H$ and all $d \in \mu(h)$, $d \in D_h$. 
		\end{enumerate}
	\end{definition}	
	The next remark follows from the definition of stability. 
\begin{remark}
		Every stable matching is individually rational.
\end{remark}

	\subsection{Algorithm} 
	In this section we present a well-known algorithm called doctor proposing deferred acceptance algorithm (DPDA). This algorithm was introduced by Gale and Shapley\cite{gale1962college}. \footnote{See Knuth\cite{knuth1976stables}, Gusfield and Irving\cite{gusfield1989marriage}, Roth and Sotomayer\cite{roth1992two}, Aldershof and Carducci\cite{aldershof1996couple} for additional results on stable matching problem in two sided matching.} Our proofs for the  existence of stable matchings use a modification of DPDA. In what follows, we give a very short description of this algorithm. Take a profile $P$. Then, the DPDA algorithm at $P$ goes as follows. 
	
	\medskip
	\noindent \emph{DPDA}:  In step 1 of the algorithm, all doctors simultaneously propose to their most
	preferred hospitals. Each hospital $h \in H$  provisionally accepts the most preferred doctors according to $P_h$. If a hospital receives more than $\kappa_{h}$ proposals, then it rejects all the doctors which do not belong to its $\kappa_h$ most preferred doctors. In any
	step $k$, the unmatched doctors propose to their most preferred hospital
	from the remaining set of hospitals who have not rejected them in any of the earlier steps. In any step of DPDA, since any hospital $h \in H$ accepts the most preferred collection of doctors according to $P_h$, it may reject some doctors that it had provisionally accepted earlier.  Hospitals whose provisional list of accepted doctors is less than their maximum capacity can still add to their accepted
	list if they receive fresh proposals. Thus the algorithm terminates when each doctor is
	matched with some hospital or has been rejected by all acceptable hospitals.
	
	\begin{remark}
		Note that in DPDA, each individual doctor proposes according to his/her individual preference. Thus, couples' preferences do not play any role in this algorithm.
	\end{remark}
	
	It is well-known that the outcome of DPDA is optimal for doctors. That is, some doctor is worse off at every other stable matching. Moreover, by responsiveness and the structure of DPDA, it follows that the outcome of DPDA is individually rational.

	The following remark follows directly from the definition of DPDA.
	
	\begin{remark}\label{ton}
		Let $\mu$ be the outcome of DPDA. Let $d \notin \mu(h)$ for some $d \in D$ and $h \in H$. Then, $h P_d \mu(d)$ implies $d' P_h d$ for all $d' \in \mu(h)$.
	\end{remark}
	
In the following lemmas, we show that the outcome of DPDA cannot be blocked by a hospital and a single doctor or by a pair of different hospitals and a couple. Some of these results are well known outcomes of DPDA, but we prove them nevertheless for the sake of completeness.
	
	\begin{lemma}\label{HS}
		The outcome of DPDA cannot be blocked by a pair $(h,s)$ for any $h \in H$ and any $s \in S$.
	\end{lemma}
	\begin{proofs}
		 Let $\mu$ be the outcome of DPDA. Assume for contradiction that $(h,s)$ blocks $\mu$ for some $h \in H$ and some $s \in S$. Since the outcome of DPDA is individually rational, $ s R_h \lambda$. Since $hP_s \mu(s)$, by the definition of DPDA and Remark \ref{ton}, either $s$ has not proposed to $h$ during the DPDA or  all the doctors in $\mu(h)$ are preferred to $s$ according to $P_h$. If $s$ has not proposed to $h$ during DPDA, then we have $\mu(s) P_s h$, a contradiction to the fact that $(h,s)$ blocks $\mu$. So, suppose $d P_h s$ for all $d \in \mu(h)$. Then, by responsiveness of hospitals' preferences, we have $\mu(h) P_h((\mu(h)\setminus \{d\})\cup s)$ for all $d \in \mu(h)$, and consequently, hospital $h$ will not block with $s$. This completes the proof of the lemma. 
	\end{proofs}

	\begin{lemma}\label{HC}
		The outcome of DPDA cannot be blocked by $((h_1,h_2),c)$ for any $h_1, h_2 \in H$ such that $h_1 \neq h_2$ and for any $c \in C$.
	\end{lemma}
	\begin{proofs}
		Let $\mu$ be the outcome of DPDA. Assume for contradiction that $\mu$ is blocked by $((h_1,h_2),c)$. Let $\mu(f)=h_f$ and $\mu(m)=h_m$. By the definition of a block,  $(h_1,h_2) P_c (h_f,h_m)$. 
		
		Suppose $h_f R_f h_1$ and $h_m R_m h_2$. Since  $(h_1,h_2) \neq (h_f,h_m)$, this means  $(h_f,h_m) P_c (h_1,h_2)$, a contradiction. 		
		Now, suppose $h_1 P_f h_f$ or $h_2 P_m h_m$. Without loss of generality, assume $h_2 P_m h_m$. Since the outcome of DPDA is individually rational, $h_2 \neq \lambda$. Because $h_2 P_m h_m$, by  Remark \ref{ton}, $m$ proposed to $h_2$ at some step of DPDA and got rejected. Since $h_1 \neq h_2$, by Lemma \ref{HS}, we have $\mu(h_2) P_{h_2}  ((\mu(h_2) \setminus \{d\}) \cup m)$ for all $d \in \mu(h_2)$. However, this contradicts the definition of a block.
	\end{proofs}

In what follows, we give a lemma which shows that the outcome of DPDA cannot be blocked by a pair of dummy hospitals and a couple.

\begin{lemma}\label{UC}
	The outcome of DPDA cannot be blocked by $((\lambda,\lambda),c)$ for any $c \in C$.
\end{lemma}
\begin{proofs}
	Let $\mu$ be the outcome of DPDA. Assume for contradiction that $\mu$ is blocked by $((\lambda,\lambda),c)$. Let$\mu(f)=h_f$ and $\mu(m)=h_m$. By the definition of block,  $(\lambda,\lambda)P_c(h_f,h_m)$. By responsiveness, this means  $\lambda P_x h_x $ for some $x \in \{f,m\}$. However, by the definition of DPDA, $\lambda P_x h_x$ means $x$ proposed to $\lambda$ before proposing to $h_x$ and got rejected. This is contradiction as, by our assumption, $\lambda$  cannot reject a doctor.
\end{proofs}

In the following lemmas, we give conditions when an outcome of DPDA cannot be blocked by a pair of same hospitals and a couple. 

\begin{lemma}\label{No1}
	Suppose $\mu$ is an outcome of DPDA. Then for any $h \in H$ and any $c=\{f,m\} \in C$,
	\begin{itemize}
		\item[(i)] $f P_h m$ and $h P_f \mu(f)$ implies $((h,h),c)$ cannot block $\mu$, and
		\item[(ii)] $m P_h f$ and $h P_m \mu(m)$ implies $((h,h),c)$ cannot block $\mu$.
	\end{itemize}
\end{lemma}
\begin{proofs}
	Assume without loss of generality that $f P_h m$ and $h P_f \mu(f)$. By the definition of DPDA, $f$ proposed to $h$ and got rejected. By Remark \ref{ton}, $d P_h f$ for all $d \in \mu(h)$. Because $f P_h m$,  $\{d,d'\} P_h \{f,m\}$ for all $d,d' \in \mu(h)$. This means $h$ cannot block $\mu$ with $c$. This completes the proof.
\end{proofs}

\begin{lemma}\label{No2}
		Suppose $\mu$ is an outcome of DPDA. Then for any $h \in H$ and any $c=\{f,m\} \in C$, $((h,h),c)$ cannot block $\mu$ if $h=\mu(x)$ for some $x \in \{f,m\}$
\end{lemma}
\begin{proofs}
	Without loss of generality, let $h=\mu(f)$. Let $((h,h),c)$ block $\mu$. 
	
	By the definition of a block, $(h,h)P_c(\mu(f),\mu(m))$. Thus, by responsiveness of couples' preferences, $h P^0_m h_m$. By the definition of DPDA, this means $m$ proposed to $h$ and got rejected. Since $f \in \mu(h)$, by Lemma \ref{HS} and the definition of block, we have $\mu(h) P_{h}  ((\mu(h) \setminus \{d\}) \cup m)$ for all $d \in \mu(h)$. This means $h$ will not block $\mu$ with $c$, which is a contradiction.
\end{proofs}

\begin{remark}\label{new1}
	If a doctor $d \in D$ and hospital $h \in H$ are each other's top ranked alternative, then trivially for a stable match, they must be matched to each other.
\end{remark}

\section{Stable matching is not guaranteed at arbitrary responsive profiles} \label{section3}
In this section, we show by means of two examples that existence of a stable matching is not guaranteed at arbitrary  responsive preference profiles. The two examples are slightly different to suit the subsequent sections.

\begin{example}\label{ex}\normalfont
	Let $H= \{h_1,h_2,h_3\}$,  $\kappa_{h}=2$ for all $h \in \{h_1,h_2,h_3\}$, and  $D=\{d_1,d_2,d_3,d_4,f,m\}$ where  $c=\{f,m\}$ is a couple. 
	
	Suppose $r_1(P_{d_1})=h_2$ and $r_2(P_{d_1})=r_1(P_{d_2})=h_1$. Further, $r_1(P_{d_3})=h_2$ and $r_1(P_{d_4})=h_3$. For the couple, suppose $h_2 P_f h_1$, $h_2 P_f h_3$ and $h_1 P_m h_3$ but $(h_1,h_1)P_c(h_2,h_3)$. Finally $h P_x \lambda$ for all $x \in \{f,m\}$ and all $h \in H$. 
	
	The above mentioned preferences along with the preferences of all the hospitals over individuals are given in Table \ref{tb1}. The preference of all hospitals over pairs of doctors can be any responsive preference over pairs of doctors. However $\{f,m\}P_{h_1} \{d_1,d_2\}$.
	
	The couple's preference over pairs of hospitals, where one member is
	matched and the other one is unmatched are assumed to be responsive and ranked below the pairs of hospitals. 
		
	\begin{table}[!hbt]
		\centering
		\begin{tabular}{c c c c c c c c c c}
			\hline
			$P_{d_1}$ & $P_{d_2}$ & $P_{d_3}$ & $P_{d_4}$& $P_f$ & $P_m$ & $P_c$ & $\bar{P}_{h_1}$ & $\bar{P}_{h_2}$ & $\bar{P}_{h_3}$\\
			\hline
			$h_2$ & $h_1$ & $h_2$ & $h_3$ & $\vdots$ & $\vdots$ & $\vdots$ & $d_3$ & $d_3$ & $d_4$ \\  
			$h_1$ & $\vdots$ & $\vdots$ & $\vdots$ & $h_2$ & $h_1$ & $(h_1,h_1)$ & $d_4$ & $d_4$ & $d_3$ \\
			$\vdots$ & & & & $\vdots$ & $\vdots$ & $\vdots$ & $f$ & $f$ & $d_1$ \\
			& & & & $h_1$ & $h_3$ & $(h_2,h_3)$ & $d_1$ & $d_1$ & $f$ \\
			& & & & $\vdots$ & $\vdots$ & $\vdots$ & $d_2$ & $m$ & $m$ \\
			& & & & & & & $m$ & $d_2$ & $d_2$\\
			\hline
		\end{tabular}
		\caption{}
		\label{tb1}
	\end{table}
	
    The couple and $h_1$ have responsive preferences.
	In what follows, we argue that there is no stable matching for the preference profile given in Table \ref{tb1}. 		
	
	Assume for contradiction, that there exists a stable matching $\mu$ for the given preference profile. By Remark \ref{new1}, it must be that $\mu(d_3)=h_2$ and $\mu(d_4)=h_3$, as the doctors and hospitals are each other's top ranked alternative. Thus, there is potentially only one vacancy to be filled in $h_2$ and $h_3$ respectively. 
	
Particularly, it is not possible for the couple to be matched to $(h_2,h_2)$ and $(h_3,h_3)$ for a stable matching. Also, since the couple prefers to be matched with any two hospitals than having at least one member of the couple unmatched and we have sufficiently many number of vacancies in all the three hospitals, thus, $\{\mu(f), \mu(m)\} \subseteq \{h_1,h_2,h_3\}$. 

Now we look at the following allocations of the couple in $\mu$.
	
	\begin{itemize}
		\item[(i)] Suppose $\mu(c)=(h_2,h_1)$. \\
		Note that $f P_{h_2} d_1$. Since $d_i \bar{P}_{h_1} m$ for $i \in \{1,2\}$, $r_2(P_{d_1})=r_1(P_{d_2})=h_1$. So, $\mu(h_2)=\{d_3,f\}$. As $\kappa_{h_1}=2$, therefore $d_i \notin \mu(h_1)$ for some $i$. Thus, stability of $\mu$ implies that $(h_1,d_i)$ blocks $\mu$.
		\item[(ii)] Suppose $\mu(c)=(h_1,h_1)$. \\
	Then $(h_1,d_2)$ blocks $\mu$.
		\item[(iii)] Suppose $\mu(c)=(h_2,h_3)$.\\
		 Since $\{f,m\} P_{h_1} \{d_1,d_2\}$ and $(h_1,h_1) P_c (h_2,h_3)$, $((h_1,h_1),c)$ blocks $\mu$.
		\item[(iv)] Suppose $\mu(c)=(h_1,h_3)$. \\
		Since $h_2 P_f h_1$, responsiveness implies $(h_2,h_3) P_c (h_1,h_3)$. This together with the fact that $f \bar{P}_{h_2} d_1 \bar{P}_{h_2} d_2$ implies $((h_2,h_3),c)$ blocks $\mu$.
		\item[(v)] Suppose $\mu(c)=(h_3,h_1)$. \\
		Since $h_2 P_f h_3$, responsiveness implies $(h_2,h_1) P_c (h_3,h_1)$. This together with the fact that $f \bar{P}_{h_2} d_1 \bar{P}_{h_2} d_2$ implies $((h_3,h_1),c)$ blocks $\mu$.
		\item[(vi)] Suppose $\mu(c)=(h,h_2)$ for some $h \in \{h_1,h_3\}$. \\
		Thus, $d_1 \notin \mu(h_2)$. Since $r_1(P_{d_1})=h_2$ and $d_1 P_{h_2} m$, stability of $\mu$ implies that $(h_2,d_1)$ blocks $\mu$.
	\end{itemize}
		Cases (i)-(vi) together are exhaustive. Thus, it follows that there is no stable matching for the given preference profile in Table \ref{tb1}.
		
\end{example}

Note that, the arguments of  this example can also be used to give an example with no stable matching when we consider $h_2$ and $h_3$ as the same hospital. Thus, we just consider two vacancies for both the hospitals while excluding $d_3$ and $d_4$ from the set of doctors. It follows from the above argument that there does not exist a stable matching when $(h_1,h_1)P_c(h_2,h_2)$ with $h_2 P_f h_1$, $h_1 P_m h_2$ and the preferences of hospitals over $\{f,m,d_1,d_2\}$ as given by $\bar{P}_{h_1}$ and $\bar{P}_{h_2}$.

\begin{example}\label{ex2}\normalfont
	Let $H= \{h_1\}$,  $\kappa_{h_1}=2$ and  $D=\{d_1,d_2,f,m\}$ where  $c=\{f,m\}$ is a couple. 
	
	The preference of the hospital over individuals and pairs of individuals, preference of individual doctors and the couple preferences are given in Table \ref{tb3}.

	The preference of $h_1$ over the pairs of doctors where at least one position is vacant is assumed to ranked below the shown pairs.
	
	\begin{table}[!hbt]
		\centering
		\begin{tabular}{c c c c c c c c}
			\hline
			$\bar{P}_{h_1}$ & $P_{d_1}$ & $P_{d_2}$ & $P_{f}$ & $P_{m}$ & $P_c$ & $P_{h_1}$\\
			\hline
			$f$ & $h_1$ & $h_1$ & $\lambda$ & $h_1$ & $(\lambda,h_1)$ & $\{f,d_1\}$ \\
			$d_1$ & $\lambda$ & $\lambda$ & $h_1$ & $\lambda$ & $(h_1,h_1)$  & $\{f,d_2\}$\\
			$d_2$ & & & & & $(\lambda,\lambda)$ & $\{f,m\}$\\
			$m$ &  &  & & & $(h_1,\lambda)$ & $\{d_1,d_2\}$ \\
			& & & & & & $\{d_1,m\}$\\ 
			& & & & & & $\{d_2,m\}$\\
			\hline
		\end{tabular}
		\caption{}
		\label{tb3}
	\end{table}
	
	Note that the couple and $h_1$ have responsive preferences
	In what follows, we argue that there is no stable matching for the preference profile given in Table \ref{tb3}. 		
	
	Assume for contradiction, that there exists a stable matching $\mu$ for the given preference profile. We look at the following allocations of the couple in $\mu$.
	
	\begin{itemize}
		\item[(i)] Suppose $\mu(c)=(\lambda,h_1)$. \\
		Since $d_i \bar{P}_{h_1} m$, $r_1(d_i) = h_1$ for $i \in \{1,2\}$, and $\kappa_{h_1}=2$, therefore $d_i \notin \mu(h_1)$ for some $i$. Thus, stability of $\mu$ implies that $(h_1,d_i)$ blocks $\mu$.
		\item[(ii)] Suppose $\mu(c)=(\lambda,\lambda)$.\\
		Since $\{f,m\} P_{h_1} \{d_1,d_2\}$ and $(h_1,h_1) P_c (\lambda,\lambda)$, $((h_1,h_1),c)$ blocks $\mu$.
		\item[(iii)] Suppose $\mu(c)=(h_1,h)$ for some $h \in \{h_1,\lambda\}$. \\
		Since $r_1(P_f) = \lambda$, it follows by responsiveness that $((\lambda,h),c)$ blocks $\mu$.
	\end{itemize}
	Since cases (i)-(iii) are exhaustive, it follows that there is no stable matching for the given preference profile in Table \ref{tb3}.

\end{example}

\section{Conditions on couples' preferences for stability} \label{sec3} 
In view of Example \ref{ex}, we look for necessary and sufficient  conditions on couples' preferences that guarantee the existence of a stable matching for every profile.
		
Let $P^0_{C}=(\{{P}^0_c \}_{c \in {C}})$ be  a given profile of preferences of the couples. Thus, for any $c=\{f,m\} \in C$, and a given couple preference $P^0_c$, $P^0_f$ and $P^0_m$ denote the individual preferences of $f$ and $m$ respectively. Since preferences of couples are responsive, a couple preference $P^0_c$ uniquely determines the individual preferences $P^0_f$ and $P^0_m$ of the members of the couple. In what follows, we present a condition on $P^0_C$ called extreme-altruism. 
		
	\begin{definition}\label{alt}
		A profile of couple preferences $P^0_C$ is said to satisfy extreme-altruism if for all $c=\{f,m\} \in C$, all $h \in H$ and all $h',h'' \in H \cup \{\lambda\}$: 
		\begin{itemize}
			\item[(i)] $h'P^0_f h$, $h'R^0_f h''$, $h'' R^0_m \lambda$ and $\kappa_h \leq |D| -2$ imply $(h',h'') P^0_c (h,h)$, and 
			\item[(ii)] $h'P^0_m h$, $h'R^0_m h''$, $h'' R^0_f \lambda$ and $\kappa_h \leq |D| -2$ imply $(h'',h') P^0_c (h,h)$.
		\end{itemize}
		
	\end{definition}

For illustration of extreme-altruism, consider a couple $c=\{f,m\}$ and hospitals $h \in H$ and $h',h'' \in H \cup \{\lambda\}$. Suppose $f$ strictly prefers $h'$ to $h$ and weakly  prefers $h'$ to $h''$, and $m$ weakly prefers $h''$ to $\lambda$. Then extreme-altruism says that, if $h$ does not have sufficiently large capacity, then $c$ ranks $(h',h'')$ above $(h,h)$. Note here, that if $h$ has a sufficiently large capacity, then we would not require any restriction on couples' preferences, simply because the hospital $h$ will have enough vacancies to accommodate at least one member of the couple, thereby removing the possibility of the couple to block with the hospital.

Out of two allocations of  a couple, one where both the members are allocated at the same hospital and the other where one member moves to a better hospital (according to his/her individual preference), the couple prefers the latter to the former. For instance, if the hospital $h'$ gives a substantially better salary compared to $h$. Then the couple would rather prefer on member to be at $h'$ than both the members of the couple to be at $h$. 

Our next lemma shows that extreme-altruism and responsiveness together imply that if two hospitals are acceptable for both members of a couple, then the ranking of those two hospitals by each member of the couple is always the same. That is, if $h_1$ and $h_2$ are both acceptable by each member of a couple  $\{f,m\}$, then $f$ and $m$ will have same ranking over $h_1$ and $h_2$.

\begin{lemma}\label{same}
	Let $P^0_C$ be a profile of couple preferences satisfying extreme-altruism. Then, for any $c =\{f,m\}$ and any $h_1, h_2 \in H$ such that $h P^0_x \lambda$ for all $x \in \{f,m\}$ and all $h \in \{h_1,h_2\}$, we have $h_1 P^0_f h_2$ if and only if $h_1 P^0_m h_2$. 
\end{lemma}

\begin{proofs}	
	Let $P^0_C$ be a profile of couple preferences satisfying extreme-altruism. Consider a couple $c=\{f,m\}$ with preference $P^0_c$ and consider two  hospitals $h_1,h_2 \in H$. Assume for contradiction  $h_1 P^0_f h_2 P^0_f \lambda$ and $h_2 P^0_m h_1 P^0_m \lambda$. By responsiveness and Condition (i)  of the definition of extreme-altruism, we have  $(h_1,\lambda) P^0_c (h_2,h_2) P^0_c(\lambda,h_2)$. Again, responsiveness and Condition (ii) of extreme-altruism implies $(\lambda,h_2) P^0_c (h_1,h_1) P^0_c(h_1,\lambda)$. However, this is a contradiction. 
	\end{proofs}

Our next corollary follows directly from Lemma \ref{same}. It says the following. Suppose $P^0_C$ satisfies the extreme-altruism. Consider a couple. Suppose that every hospital is acceptable for each member of the couple. Then, the members of the couple have the same individual preference over $H$.

\begin{cor}
	Let $P^0_C$ be a profile of couple preferences satisfying  extreme-altruism. Let $c=\{f,m\} \in C$ be a couple such that $h P^0_f \lambda$ and $h P^0_m \lambda$ for all $h \in H$. Then $P^0_m = P^0_f$. 
\end{cor}
	
	For a profile of preferences $P^0_{C}$ of the couples, an extension of $P^0_{C}$ is defined as a preference profile $P=(\{P_d\}_{d \in D}, \{P_c\}_{c \in C},\{P_h\}_{h \in H})$ such that $P_c = P^0_c$ for all $c \in C$. 
		
The following theorem says that extreme-altruism of $P^0_C$ is necessary and sufficient for the existence of a stable matching for \emph{every} extension of $P^0_{C}$. 
		
\begin{theorem}\label{theo3}\begin{itemize}
		\item[(i)]  If $P^0_{C}$ satisfies extreme-altruism, then a stable matching exists for any extension  of $P^0_{C}$.
		\item[(ii)] If $P^0_{C}$ does not satisfy extreme-altruism for all $h \in H$, then there exists an extension of $P^0_{C}$ with no stable matching.			
		\end{itemize}
\end{theorem}
		
\begin{proof}[Part (i)] Consider a preference profile $P$ that is an arbitrary extension of $P^0_{C}$ where $P^0_{C}$ satisfies extreme-altruism. We show that the DPDA where each member of each couple proposes according to his/her individual preference gives a stable matching for $P$. 

Let $\mu$ be the outcome. Suppose $\mu$ is not stable at $P$. Since DPDA is individually rational,  Lemma \ref{HS}, Lemma \ref{HC} and Lemma \ref{UC} imply that $\mu$ is blocked by $((h,h),c)$ for some $h \in H$ and some $c=\{f,m\} \in C$. Let $\mu(f)=h_f$ and $\mu(m)=h_m$. Thus $(h,h)P^0_c(h_f,h_m)$. 
	
Assume without loss of generality that $f P_h m$. By Lemma \ref{No1} and Lemma \ref{No2}, we know that $h_f P^0_f h$. Thus, by responsiveness, we must have $h P^0_m h_m$. Suppose $\kappa_h > |D| -2$. Then $\{f,m\} \notin \mu(h)$ implies $|\mu(h)| \leq |D|-2 < \kappa_h$. But since $h P^0_m h_m$, $m$ proposed to $h$ at an earlier step of DPDA and got rejected even when $h$ had a vacancy. Thus $m \notin D_h$ which implies that $((h,h),c)$ cannot block $\mu$. Thus $\kappa_h \leq |D| -2$.

By the definition of DPDA, we have $h_f R^0_f \lambda$ and $h_m R^0_m \lambda$. Also, we know that $x$ is weakly preferred to itself. This together with extreme-altruism implies $(h_f,\lambda)P^0_c(h,h)P^0_c(h_f,h_m)$. 

This contradicts the fact that $h_m R^0_m \lambda$. Thus $(h_f,h_m)R^0_c(h_f,\lambda)P^0_c(h,h)$ which contradicts that $((h,h),c)$ blocks $\mu$. This completes the proof of part (i) of Theorem \ref{theo3}. 

\medskip

	\noindent [Part (ii)] Suppose $P^0_C$ does not satisfy extreme-altruism. We show that there exists an extension of $P^0_C$ with no stable matching. 
	
	Since $P^0_C$ does not satisfy extreme-altruism, there is a couple $c=\{f,m\}$, a hospital $h$ such that $\kappa_h \leq |D|-2$ and hospitals $h_1,h_2 \in H \cup \{\lambda\} \setminus \{h\}$  such that:
	\begin{itemize}
		\item[(i)] either, $h_1 P^0_f h$, $h_1 R^0_f h_2$ and $h_2 R^0_m \lambda$, but $(h,h) P^0_c (h_1,h_2)$, 
		\item[(ii)] or, $h_1 P^0_m h$, $h_1 R^0_m h_2$ and $h_2 R^0_f \lambda$, but $(h,h) P^0_c (h_2,h_1)$.
	\end{itemize}
    Assume without loss of generality that (i) holds. 
    
	By responsiveness, $h_1 P^0_f h$ and $(h,h) P^0_c (h_1,h_2)$, implies $h P^0_m h_2$. Consider a preference profile $P$ such that 
	\begin{enumerate}[(i)]
		\item for all $h' \in H \setminus \{h,h_1,h_2\}$,  either $|\{d:d P_{h} c \mbox{ and } r_1(P_d)=h\}|=\kappa_{h'}$ or $f,m \notin D_{h'}$,
		\item there are doctors $d_1,d_2 \in D \setminus \{f,m\}$ such that $f P_h d_1 P_h d_2 P_h m$ and $\{f,m\} P_h \{d_1,d_2\}$,
		\item  $|\{d:d P_{h} f \mbox{ and } r_1(P_d)=h\}|=\kappa_{h}-2$. This is possible since $\kappa_h \leq |D|-2$ implies $\kappa_h -2 \leq |D|-4$,
		\item either $f P_{h_1} d_1 P_{h_1} m$ and $|\{d:d P_{h_1} f \mbox{ and } r_1(P_d)=h_1\}|=\kappa_{h_1} - 1$, or $r_1(P_{h_1})=f$, $r_2(P_{h_1})=d_1$ and $m \notin D_{h_1}$,
		\item either $|\{d:d P_{h_2} c \mbox{ and } r_1(P_d)=h_2\}|=\kappa_{h_2} -1$ or $f \notin D_{h_2}$,
		\item  $r_1(P_{d_1})=h_1$ and $r_2(P_{d_1})=r_1(P_{d_2})=h$.	
	\end{enumerate}
	
	But it trivially follows from this preference profile that for a stable matching $\mu$, $h' \notin \{\mu(f),\mu(m)\}$ for $h' \notin \{h,h_1,h_2\}$. Also, it is not possible that $\mu(c)=(h_i,h_i)$ for $i \in \{1,2\}$. Thus, by our construction, a stable matching exists for this matching problem if and only if there is a stable matching for Example \ref{ex}. However, since there does not exist a stable matching for Example \ref{ex}, thus, we do not have a stable matching for $P$. This completes the proof of part (ii) of Theorem \ref{theo3}.
\end{proof}

From the above example, it is clear that even if all hospitals view all the doctors as acceptable, violation of extreme-altruism can lead to a preference profile with no stable matching, if we have enough doctors to fulfil the capacity constraints as given by points (i),(ii), (iv) and (v). Thus we get the following corollary.

\begin{cor}
	Suppose $D_h=D$ for all $h \in H$. Moreover, $\sum_{h \in H} \kappa_h = |D|$, then:
	\begin{itemize}
		\item[(i)]  If $P^0_{C}$ satisfies extreme-altruism, then a stable matching exists for any extension  of $P^0_{C}$.
		\item[(ii)] If $P^0_{C}$ does not satisfy extreme-altruism for all $h \in H$, then there exists an extension of $P^0_{C}$ with no stable matching.			
	\end{itemize}
\end{cor}

\section{Conditions on hospitals' preferences for stability} \label{sec4}
In Section \ref{sec3}, we have discussed a necessary and sufficient condition on couples' preferences that guarantees the existence of a stable matching for every collection of preferences of the hospitals. In this section, we look at the other side of the problem, that is, we look for  necessary and  sufficient condition on hospitals' preferences so that a stable matching exists for every collection of preferences of the doctors (both individuals and couples). 
	
Let $P^0_{H}=(\{P^0_h\}_{h \in H})$ be  a given profile of preferences of the hospitals. In what follows, we introduce the \emph{aversion to couple diversity} property. 
	
\begin{definition}
	A profile of hospital preferences $P^0_H$ is said to have aversion to couple diversity  if for all $h \in H$, all $c=\{f,m\}$ and all  $d_1,d_2 \in D$ with $f,m \in D_h$ such that either (i) $f P^0_h d_1 P^0_h d_2 P^0_h m$ and  $|\{d:d P^0_h m\}| > \kappa_h$, or (ii) $m P^0_h d_1 P^0_h d_2 P^0_h f$ and  $|\{d:d P^0_h f\}| > \kappa_h$, we have $\{d_1,d_2\} P^0_h \{f,m\}$. 
\end{definition}

Consider a hospital $h$ with a preference $P_h$ over acceptable and feasible sets of doctors. Take a couple $c=\{f,m\}$ such that both $f$ and $m$ are acceptable for $h$ but at least one of them is not  amongst the top-$\kappa_h$ doctors according to the restriction of $P_h$ over individual doctors. Suppose that  there are two doctors $d_1,d_2$ who are ranked in-between $f$ and $m$ according to $P_h$. Then, aversion to couple diversity says that the set  $\{d_1,d_2\}$ must be preferred to the couple $c$ according to $P_h$. Note here that if both the members of the couple are in the top-$\kappa_h$ doctors according to the restriction of $P_h$ over individual doctors, then we do not need this condition as for a stable matching, the couple will always be a part of $h$.

So, in other words, whenever a hospital compares a couple and another pair of doctors over which responsiveness does not induce the comparison, the hospital prefers the couple only if at most one doctor from the other pair ranks in-between the members of the couple. Thus, a hospital has aversion to couple diversity if it does not like to employ a couple whose members have relatively more dissimilar ranking in its preference. 

It is important to note here that the diversity aversion just applies to couples and not single doctors as two single doctors can not apply to a hospital together and block a matching. On the other hand, we can encounter a situation where a couple applies to a hospital such that a member of a couple is individually worse off but the couple is better off as a whole.

For a profile of preferences $P^0_{H}$ of the hospitals, an extension of $P^0_{H}$ is defined as a preference profile $P=(\{P_d\}_{d \in D}, \{P_c\}_{c \in C},\{P_h\}_{h \in H}  )$ such that $P_h = P^0_h$ for all $h \in H$. 

Our next theorem says that the aversion to couple diversity of $P^0_H$ is  necessary and sufficient  for the existence of a stable matching for every extension of $P^0_H$. 

\begin{theorem}\label{hosp}	
	\begin{itemize}
		\item[(i)]	If $P^0_{H}$ satisfies aversion to couple diversity property, then a stable matching exists at every extension  of $P^0_{H}$.
		
		\item[(ii)] If $P^0_{H}$ does not satisfy aversion to couple diversity property, then there exists an extension of $P^0_{H}$ with no stable matching.
	\end{itemize}		 
\end{theorem}
\begin{proof} [Part (i)]   Consider a preference profile $P$ that is an arbitrary extension of $P^0_{H}$ where $P^0_{H}$ satisfies aversion to couple diversity. We show that the DPDA where each member of each couple proposes according to his/her individual preference gives a stable matching for $P$. 
	
Let $\mu$ be the outcome. Suppose $\mu$ is not stable at $P$. Since DPDA is individually rational,  Lemma \ref{HS}, Lemma \ref{HC} and Lemma \ref{UC} imply that $\mu$ is blocked by $((h,h),c)$ for some $h \in H$ and some $c=\{f,m\} \in C$. Let $\mu(f)=h_f$ and $\mu(m)=h_m$. Thus $(h,h)P^0_c(h_f,h_m)$. 
	
	Assume without loss of generality that $f P_h m$. If $m \notin D_h$, then $((h,h),c)$ cannot block $\mu$ as it violates individual rationality. Thus $m P^0_h \lambda$. By Lemma \ref{No1} and Lemma \ref{No2}, we know that $h_f P^0_f h$. Thus, responsiveness implies $h P_m h_m$. 
	
	It follows that, before applying to $h_m$,  $m$ applied to $h$ and got rejected. Therefore, $|\{d: d P^0_h m  \mbox{ and } d\in \mu(h)\}| =\kappa_h$. Since $f P^0_h m$ and $f \notin \mu(h)$, we have $|\{d:d P^0_h f\}| > \kappa_h$. This, together with aversion to couple diversity, implies that $\{d,d'\}P^0_h \{f,m\}$ for all $d,d' \in  \mu(h)$, which is a contradiction to the fact that $((h,h),c)$ blocks $\mu$.  This completes the proof of part (i) of Theorem \ref{hosp}. 
	
	\medskip

	\noindent[Part (ii)] Suppose $P^0_H$ does not satisfy the aversion to couple diversity  property. We show that there is an extension of $P^0_{H}$ with no stable matching. 
	
	Since $P^0_H$ does not satisfy the aversion to couple diversity  property, we have $h \in H$,  $c=\{f,m\} \in C$ and  $d_1,d_2 \in D$ with $f,m \in D_h$ such that  $ \{f,m\} P^0_h \{d_1,d_2\}$ and either  \\
	(i) $f P^0_h d_1 P^0_h d_2 P^0_h m$ and  $|\{d:d P^0_h m\}| > \kappa_h$,  or \\
	(ii) $m P^0_h d_1 P^0_h d_2 P^0_h f$ and  $|\{d:d P^0_h f\}| > \kappa_h$.

	Assume without loss of generality that (i) holds. Consider a preference profile $P$ such that
	\begin{itemize}
		\item[(i)] $\lambda P_f h$,
		\item[(ii)] $h P_m \lambda$,
		\item[(iii)] $(h,h)P_c(\lambda,\lambda)$,
		\item[(iv)] for all $h'\in H \setminus \{h\}$ either $h P_f h'$, or $|\{d:dP^0_{h'} f \mbox{ and } r_1(P_d)=h'\}|=\kappa_{h'}$,
		\item[(v)] for all $h'\in H \setminus \{h\}$ either $\lambda P_m h'$, or $|\{d:dP^0_{h'} m \mbox{ and } r_1(P_d)=h'\}|=\kappa_{h'}$, 
        \item[(vi)] $r_1(P_{d_1})=r_1(P_{d_2})=h$, and
		\item[(vii)]  $|\{d:dP^0_h m \mbox{ and } r_1(P_d)=h\}|=\kappa_h$. Note, that this also includes $d_1$ and $d_2$.  
	\end{itemize}

	But it trivially follows from the given preference profile that for a stable matching $\mu$, for all $h' \in H \setminus \{h\}$, $h' \notin \{\mu(f),\mu(m)\}$. Thus, by our construction, a stable matching exists for this matching problem if and only if there is a stable matching for Example \ref{ex2}. However, since there does not exist a stable matching for Example \ref{ex2}, thus we do not have a stable matching for $P$. This completes the proof of part (ii) of Theorem \ref{hosp}.
\end{proof}

From the example above, it is not clear if non-aversion to couple diversity can always lead to a preference profile with no stable matching when doctors prefer to be matched any hospital than being unemployed. We show by the means of an example that if all the doctors are averse to unemployment, then we can not always obtain an extension of $P^0_H$ with no stable matching, when $P^0_H$ does not satisfy aversion to couple diversity.

\begin{example}\normalfont
	Consider a matching problem with $H=\{h_1,h_2\}$ and $D=\{f,m,s_1,s_2\}$ such that $c=\{f,m\}$ is the only couple. Let $\kappa_{h_1}=\kappa_{h_2}=2$. Thus, $\sum_{h \in H} \kappa_h =|D|$. The preferences of hospitals on individual doctors is given in the table below. The doctors prefer to be matched to any hospital than being unemployed.
	
	\begin{table}[!hbt]
		\centering
		\begin{tabular}{c c}
			\hline
			$P_{h_1}$ & $P_{h_2}$\\
			\hline
			$f$ & $s_2$\\
			$s_1$ & $m$\\
			$s_2$ & $f$\\
			$m$ & $s_1$\\
			\hline
		\end{tabular}
		\label{tb4}
	\end{table}
	\noindent Let $\{f,m\} P_{h_1} \{s_1,s_2\}$. Thus the preference $h_1$ does not follow aversion to couple diversity.  
	
	We show that there exists a stable matching for these preferences of hospitals for any preferences of the doctors and the couple. 
	
	Let $\mu$ be a matching for the given preferences of the hospitals such that $\mu(h_1)=\{f,s_1\}$ and $\mu(h_2)=\{s_2,m\}$. Clearly, $h_1$ and $h_2$ have their top ranked doctors. Thus, neither $h_1$ nor $h_2$ would like to block $\mu$ with any other doctor. Also, no doctor would block $\mu$ with $\lambda$ as all doctors prefer being matched to any hospital than be unemployed. Thus $\mu$ is stable for any preferences of doctors even when there is no aversion for diverse couples.
\end{example}

The above example leads to the following corollary.

\begin{cor}
	Suppose $h P_d \lambda$ for all $h \in H$ and all $d \in D$, then a stable matching always exists for any extension of $P^0_H$ when $P^0_H$ satisfies aversion to diverse couples.
\end{cor}

\section{Concluding remarks} \label{sec5}
As we have discussed earlier, existence of a stable matching is guaranteed at a preference profile if it satisfies responsiveness as defined in Klaus and Klijn\cite{klaus2005responsive}. The different result in our paper stems from the fact that we allow for a setwise blocking notion, which allows for the existence of more blocking coalitions. Here, a hospital is allowed to replace two doctors by a couple, whilst Klaus and Klijn consider pairwise blocking. Thus, if a couple wishes to block with a hospital, both the members of the couple will be considered separately by the hospital instead of considering the couple as a whole. We explain this in detail in the following paragraph.

Consider Example \ref{ex} and Example \ref{ex2}. According to the model in Klaus and Klijn\cite{klaus2005responsive}, given the preference $\bar{P}_{h_1}$, the pair $\{f,m\}$ can not block with $h_1$ to remove the pair of doctors $ \{d_1,d_2\}$. By their blocking notion, each member of the couple can only replace a doctor who is ranked lower to that member of the couple.  The fact that $d_1$ and $d_2$ are ranked in-between  $f$ and $m$, prevent the couple to block with $h_1$.

Now, consider the matching  $\mu_1$ for Example \ref{ex} and $\mu_2$ for Example \ref{ex2} such that $\mu_1(c)=(h_2,h_3)$, $\mu_1(d_1)=\mu_1(d_2)=h_1$ and $\mu_2(c)=(\lambda,\lambda)$, $\mu_2(d_1)=\mu_2(d_2)=h_1$. Note that, in our model, both these matchings are blocked by $((h_1,h_1),c)$. However, this block is \emph{not} possible according to the model in Klaus and Klijn\cite{klaus2005responsive}. It can be verified that $\mu_1$ and $\mu_2$ are indeed stable according to their  model.

In this paper, we have shown that the existence of a stable matching is not guaranteed when couples and/or hospitals have complete and responsive preferences. We have  provided (a) necessary and sufficient conditions on couples' preferences  so that a stable matching exists at every extension of those preferences, and (b) necessary and sufficient conditions on hospitals' preferences  so that a stable matching exists at every extension of those preferences. 
Additionally, we have provided algorithms that produce a stable matching whenever that exists in this framework.

	\bibliography{mybib}

	\bibliographystyle{amsplain}

	\end{document}